\pdfoutput=1
\documentclass[10pt,pra,aps,superscriptaddress,twocolumn]{revtex4-1}
\usepackage{amsmath,bbm}
\usepackage{amsthm}
\usepackage{amsmath}
\usepackage{latexsym}
\usepackage{amssymb}
\usepackage{float}
\usepackage{graphicx}           
\usepackage{color}
\usepackage{mathpazo}
\usepackage{comment}
\usepackage{enumitem}
\usepackage{multirow}
\usepackage{setspace}
\usepackage{ulem}
\usepackage[colorlinks=true,linkcolor=blue,citecolor=magenta,urlcolor=blue]{hyperref}
\allowdisplaybreaks

\newcommand{\be}{\begin{equation}}
\newcommand{\ee}{\end{equation}}
\newcommand{\bea}{\begin{eqnarray}}
\newcommand{\eea}{\end{eqnarray}}

\def\squareforqed{\hbox{\rlap{$\sqcap$}$\sqcup$}}
\def\qed{\ifmmode\squareforqed\else{\unskip\nobreak\hfil
\penalty50\hskip1em\null\nobreak\hfil\squareforqed
\parfillskip=0pt\finalhyphendemerits=0\endgraf}\fi}
\def\endenv{\ifmmode\;\else{\unskip\nobreak\hfil
\penalty50\hskip1em\null\nobreak\hfil\;
\parfillskip=0pt\finalhyphendemerits=0\endgraf}\fi}

\newcommand{\tr}{\text{Tr}}
\newcommand{\I}{\mathbbm{1}}

\newcommand{\od}{\Tilde{d}}

\makeatletter
\newtheorem*{rep@theorem}{\rep@title}
\newcommand{\newreptheorem}[2]{%
\newenvironment{rep#1}[1]{%
 \def\rep@title{#2 \ref{##1}}%
 \begin{rep@theorem}}%
 {\end{rep@theorem}}}
\makeatother

\newreptheorem{theorem}{Theorem}
\newtheorem{lemma}{Lemma}

\newtheorem{fact}{Fact}

\newtheorem{result}{Result}

\begin{document}


\title{Measurement incompatibility and quantum advantage in communication}

\author{Debashis Saha}
\affiliation{S. N. Bose National Centre for Basic Sciences, Block JD, Sector III, Salt Lake, Kolkata 700106, India}
\affiliation{School of Physics, Indian Institute of Science Education and Research Thiruvananthapuram, Kerala 695551, India}
\author{Debarshi Das}
\affiliation{S. N. Bose National Centre for Basic Sciences, Block JD, Sector III, Salt Lake, Kolkata 700106, India}
\affiliation{Department of Physics and Astronomy, University College London, Gower Street, WC1E 6BT London, England, United Kingdom}

\author{Arun Kumar Das}
\affiliation{S. N. Bose National Centre for Basic Sciences, Block JD, Sector III, Salt Lake, Kolkata 700106, India}
\author{Bihalan Bhattacharya}
\affiliation{S. N. Bose National Centre for Basic Sciences, Block JD, Sector III, Salt Lake, Kolkata 700106, India}
\affiliation{Department of Mathematical Sciences, Indian Institute of Science Education \& Research (IISER) Berhampur,
Transit Campus, Govt. ITI, NH 59, Berhampur 760 010, Ganjam, Odisha, India}
\author{A. S. Majumdar}
\affiliation{S. N. Bose National Centre for Basic Sciences, Block JD, Sector III, Salt Lake, Kolkata 700106, India}


\begin{abstract}
  Measurement incompatibility stipulates the existence of quantum measurements that cannot be carried out simultaneously on single systems. We show that the set of input-output probabilities obtained from $d$-dimensional classical systems assisted with shared randomness is the same as the set obtained from $d$-dimensional quantum strategies restricted to compatible measurements with shared randomness in any communication scenario. Thus, measurement incompatibility is necessary for quantum advantage in communication, and \textit{any} quantum advantage (with or without shared randomness) in communication acts as a witness to the incompatibility of the measurements at the receiver's end in a semi-device-independent way. We introduce a class of communication tasks - a general version of random access codes -  to witness incompatibility of an arbitrary number of quantum measurements with arbitrary outcomes acting on $d$-dimensional systems, and provide generic upper bounds on the success metric of these tasks for compatible measurements. We identify all sets of three incompatible rank-one projective qubit measurements that random access codes can witness. Finally, we present the generic relationship between different sets of probability distributions - classical, quantum with or without shared randomness, and quantum restricted to compatible measurements with or without shared randomness - produced in communication scenarios.
\end{abstract}

\maketitle


\section{Introduction} 

In  quantum theory, a set of quantum measurements is called incompatible if these measurements cannot be performed simultaneously on a single copy of a quantum system \cite{review}. The best known example of quantum incompatibility 
pertains to the measurements of position and momentum of a quantum mechanical particle that cannot be measured simultaneously with arbitrary precision.  The notion of measurement incompatibility is an inherent property of quantum theory that differentiates it from  classical physics. 
Quantum measurement incompatibility is at the root of demonstrating various fundamental quantum aspects such as, Bell-nonlocality \cite{Fine1982,PRL103230402}, Einstein–Podolsky–Rosen steering \cite{PhysRevLett.113.160403,PhysRevLett.113.160402,tp1,PhysRevLett.115.230402,tp2}, measurement uncertainty
relations \cite{mur, MDM17, saha2020}, quantum contextuality \cite{LIANG20111,Xu2019}, quantum violation of macrorealism \cite{MM16,PhysRevA.100.042117}, and temporal
and channel steering \cite{Karthik:15,PhysRevA.91.062124,PhysRevA.97.032301}. 

 Bell-violation is the most compelling operational witness of incompatible measurements since it relies only on the input-output statistics of bipartite systems \cite{PRL103230402,PhysRevA.87.052125,math4030052}. Measurement incompatibility can also be witnessed through Einstein–Podolsky–Rosen steering  \cite{PhysRevLett.113.160403,PhysRevLett.113.160402,PhysRevLett.115.230402,PhysRevLett.116.240401,sarkar}. These protocols, however, rely on entanglement. Recently, witnessing of quantum measurement incompatibility in the prepare-and-measure scenario based on a state discrimination task  \cite{PhysRevLett.122.130402} has been proposed. It is particularly noteworthy that measurement incompatibility is necessary but not sufficient for Bell-violations employing fully untrusted devices \cite{Bene_2018,PhysRevA.97.012129}, whereas incompatibility is shown to be necessary as well as sufficient in steering with  one-sided trusted devices \cite{PhysRevLett.113.160402,PhysRevLett.115.230402}, and in state discrimination tasks with fully trusted preparations  \cite{PhysRevLett.122.130403} (see also  \cite{PhysRevA.98.012126,PhysRevA.100.042308,ioannou}). 

Nonetheless, the generic link between measurement incompatibility and nonclassical correlations in the simplest prepare-and-measure scenario is still not fully explored. The present article is motivated towards filling this 
crucial gap in the  literature. Moreover, the results presented here address the issue as to whether incompatible quantum measurements are necessary for probing quantum advantage in any communication scenario. Apart from answering this fundamental question in the affirmative, this work further aims to provide an operational witness of incompatibility for any set of quantum measurements of an arbitrary setting - 
any set of arbitrary number of measurements acting on an arbitrary but finite dimension wherein different measurements have different arbitrary number of outcomes. 


Specifically, we consider the communication scenario consisting of two players, say, Alice (sender) and Bob (receiver). 
Alice and Bob are given inputs such that each player does not know the input of the other player.  Alice, upon receiving her input, sends classical or quantum communication to Bob. Bob, upon receiving his input  and the communication sent by Alice, produces the outputs. In this scenario, we show that the input-output statistics obtained from $d$-dimensional classical systems assisted with unlimited shared randomness are the same as those obtained from the $d$-dimensional quantum states and compatible quantum measurements with unlimited shared randomness.
Therefore, any quantum advantage (with or without shared randomness) in communication task in the prepare-and-measure scenario serves as a tool to witness measurement incompatibility in a semi-device-independent way - assuming nothing about the quantum states and measurements except their dimension. Furthermore, we point out that whenever the figure of merit of any task is a convex function of the input-output statistics, its maximum value in classical communication is the same as in quantum communication with compatible measurements. 

Subsequently, we focus on Random Access Codes (RACs), a specific quantum communication task in the prepare-and-measure scenario. Based on the operational figure of merit of this task, we propose a witness of measurement incompatibility of a set of arbitrary number of quantum measurements having arbitrary number of outcomes acting on an arbitrary but finite dimensional state. We derive upper bounds (or, exact values in specific cases) of the average success probability of RAC assisted by the best classical strategy, or equivalently the best quantum strategy involving compatible measurements by the receiver. It follows that, given any set of quantum measurements, if the average success probability of RAC involving the given measurements by the receiver exceeds the above bound, then we can certify that the given measurements are incompatible. Here, it should be noted that RACs, being one of the fundamental quantum communication protocols, have been implemented in a series of experiments \cite{Spekkens2009,exprac2,armin2015,racexpt3,distributedrac,parallelrac}. Hence, the results presented in this work reveal the unrecognized fact that measurement incompatibility has been witnessed in these experiments.
Further, we identify all sets of three incompatible rank-one projective qubit measurements that can be witnessed by RACs. 

Finally, we present the generic relationship between different types of input-output statistics - classical, quantum with shared randomness, quantum without shared randomness, quantum restricted to compatible measurements with shared randomness, and quantum restricted to compatible measurements without any shared randomness.


We arrange the rest of this paper as follows. In the next Sec. \ref{sec2}, we introduce the formal definitions of incompatible quantum measurements. In Sec. \ref{sec3}, we show that incompatible measurements by the receiver are necessary for quantum advantage in any communication task. Next, considering RACs, we present an operational witness of  measurement incompatibility of an arbitrary set of quantum measurements in Sec. \ref{sec4}. Relationships between different types of probability distributions in communication tasks are depicted in Sec. \ref{sec5}. Finally, we conclude with a short discussion in Sec. \ref{sec6}.


\section{Quantum Measurement Incompatibility}\label{sec2}

An arbitrary measurement is conceptualized by some Positive Operator-Valued Measure (POVM) defined as  $E_y \equiv \{M_{b_y|y}\}_{b_y}$ with $M_{b_y|y} \geqslant 0$ for all $b_y$ and $\sum_{b_y} M_{b_y|y} = \I$. Here, $y$ corresponds to the choice of measurement, and $b_y$ denotes the outcomes of measurement $y$. 

A set of measurements, $\{E_y\}_{y}$, with $y \in [n]$ (here we use the notation $[k]:=\{1,\dots,k\}$), is compatible \cite{review} 
if there exists a parent POVM $\{G_{\kappa}: G_{\kappa} \geqslant 0 \, \forall \kappa, \, \sum_{\kappa} G_\kappa = \I \}$ and classical post-processing for each $y$ given by $\{P_y(b_y|\kappa)\}$ such that
\be \label{com}
\forall b_y,y, \quad M_{b_y|y} = \sum_\kappa P_y(b_y|\kappa) G_\kappa .
\ee 
Post-processing for each $y$ is defined by $\{P_y(b_y|\kappa)\}$ such that
\be 
P_y(b_y|\kappa) \geqslant 0 \, \, \forall y, b_y, \kappa ; \, \, \, \,  \sum_{b_y} P_y(b_y|\kappa) = 1 \, \, \forall y, \kappa.
\ee 
We note that the above definition of compatibility is equivalent to the existence of a parent POVM $\{G_\kappa\}$ whose appropriate marginals give rise to all the individual measurement effects $\{M_{b_y|y}\}_{b_y,y}$ \cite{review}.


\section{Incompatibility is necessary for quantum advantage in communication tasks}\label{sec3}

Now, we will show that incompatible measurements are necessary for showing quantum advantage in any communication task. Before proceeding, let us briefly describe a generic communication scenario consisting of two players - Alice and Bob. Alice and Bob are given inputs $x\in [l]$ and $y\in[n]$, respectively. Further,
initially neither player has any idea about the other player's input. Alice, upon receiving the input $x$ sends a $d$-dimensional classical or quantum system to Bob, where $d$ is finite. Bob, upon receiving the input $y$ and the message (which is $d$-dimensional classical or quantum system) sent by Alice, outputs $b_y\in [d_y]$ wherein $d_y$ is the number of outcomes for input $y$. The figure of merit of this communication task is determined by the set of probability distributions, $\{p(b_y |x,y)\}$.

In classical communication, they can use unlimited pre-shared randomness $\lambda$ with some probability distribution $\pi(\lambda)$, and, therefore, any typical probability can be expressed as
\be \label{pc}
p(b_y|x,y) = \sum_{m=1}^d \int_\lambda \pi(\lambda) p_a(m|x,\lambda)p_b(b_y|y,m,\lambda)\  d\lambda .
\ee 
Here, $\{p_a(m|x,\lambda)\},\{p_b(b_y|y,m,\lambda)\}$ are encoding and decoding functions by Alice and Bob, satisfying non-negativity and 
\be 
\sum_m p_a(m|x,\lambda) = \sum_{b_y} p_b(b_y|y,m,\lambda) =1 . 
\ee 
While in quantum communication, one can consider the set of probabilities with or without the pre-shared classical randomness. Suppose $\mathcal{B}(\mathbbm{C}^d)$ stands for the space of all operators acting on $d$ dimensional complex Hilbert space. 
In the former scenario with pre-shared randomness  
\bea \label{pq}
p(b_y|x,y) = \int_\lambda \pi(\lambda) \tr \left( \rho_{x,\lambda} M_{b_y|y,\lambda} \right) d\lambda, \nonumber \\
\rho_{x,\lambda}, \  M_{b_y|y,\lambda} \in \mathcal{B}(\mathbbm{C}^d)   ,
\eea 
where $\{\rho_{x,\lambda}\}$ is the quantum state sent by Alice upon input $x$ and random variable $\lambda$, and $\{M_{b_y|y,\lambda}\}$ is the measurement executed by Bob for his input $y$ and random variable $\lambda$.
Without the pre-shared randomness the expression of the probabilities reduces to,
\be \label{pqb}
p(b_y|x,y) = \tr (\rho_xM_{b_y|y} ), \quad \rho_x, M_{b_y|y} \in \mathcal{B}(\mathbbm{C}^d) .
\ee 
A communication scenario is specified by a set of natural numbers $l, n,$ and $\Vec{d} = (d_1, \cdots, d_n )$ such that $x \in [l]$, $y \in [n]$, $b_y \in [d_y]$. Given a scenario, we define the set of all probabilities obtainable by \textit{$d$-dimensional classical communication,} 
\be 
\mathcal{C}_d := \{p(b_y|x,y)\},
\ee 
where $p(b_y|x,y)$ is given by Eq. \eqref{pc}; the set of all probabilities in \textit{$d$-dimensional  quantum communication,}
\be 
\mathcal{Q}_d := \{p(b_y|x,y)\}
\ee 
where $p(b_y|x,y)$ is given by Eq. \eqref{pq}; and the set of all probabilities in \textit{$d$-dimensional quantum communication without randomness,} 
\be 
\mathcal{\overline{Q}}_d := \{p(b_y|x,y)\}
\ee 
where $p(b_y|x,y)$ is given by Eq. \eqref{pqb}.
In this work, we are interested in another two sets of probabilities. First, the set of all probabilities in \textit{$d$-dimensional  quantum communication restricted to compatible measurements,}
\be \label{QCd}
\mathcal{Q}^C_d := \{p(b_y|x,y)\}, 
\ee 
where $p(b_y|x,y)$ is given by Eq. \eqref{pq} such that the set of measurements acting on $d$-dimensional quantum states used by Bob $\{M_{b_y|y,\lambda}\}$ is compatible, that is, there exists parent POVM $\{G_{\kappa}\}$ such that
\be \label{comMlambda}
\forall b_y,y,\lambda, \quad M_{b_y|y,\lambda} = \sum_{\kappa} P_{y,\lambda}(b_y|\kappa) \ G_{\kappa}
\ee 
and
\be 
P_{y,\lambda}(b_y|\kappa) \geqslant 0 \, \, \forall y, \lambda, b_y, \kappa ; \, \, \, \,  \sum_{b_y} P_{y,\lambda}(b_y|\kappa) = 1 \, \, \forall y, \lambda, \kappa.
\ee
And second, the set of all probabilities in \textit{$d$-dimensional quantum communication restricted to compatible measurements without shared randomness,}
\be 
\mathcal{\overline{Q}}^C_d := \{p(b_y|x,y)\}, 
\ee 
where $p(b_y|x,y)$ is given by Eq. \eqref{pqb}, such that the set of measurements $\{M_{b_y|y}\}$ is compatible according to \eqref{com}. 
\begin{result}\label{result:QcC}
 Given any scenario specified by $(l,n,\Vec{d} )$,
\be
\mathcal{\overline{Q}}_d^C \subseteq \mathcal{Q}^C_d = \mathcal{C}_d .
\ee 
Thus, measurement incompatibility is necessary for any quantum advantage (with or without shared randomness) over classical communication.
\end{result}
\begin{proof}
The first relation, $\mathcal{\overline{Q}}_d^C \subseteq \mathcal{Q}^C_d$, follows from the definition of these two sets. The nontrivial part of this result is proving the equality.

Consider the case where Bob performs POVM measurements $\{G_{\kappa}\}$, which is the parent POVM of the measurement set $\{M_{b_y|y,\lambda}\}$. The 
Frenkel-Weiner theorem \cite{Frenkel2015} implies that the set of input-output probabilities $p(\kappa|x)$ with a single quantum measurement on $d$-dimensional quantum states can always be reproduced by a suitable classical $d$-dimensional communication in the presence of shared randomness. In other words, $\forall \rho_{x,\lambda}$ there exists classical strategy $\tilde{\pi}(\tilde{\lambda}),p_a(m|x,\lambda,\tilde{\lambda}),p_b(\kappa|m,\tilde{\lambda})$ such that
\be  \label{rxgk}
\tr (\rho_{x,\lambda} G_{\kappa}) = \sum_{m=1}^d \int_{\tilde{\lambda}} \tilde{\pi}(\tilde{\lambda}) p_a(m|x,\lambda,\tilde{\lambda})p_b(\kappa|m,\tilde{\lambda}) \ d\tilde{\lambda} .
\ee
Here, note that, in the scenario considered by Frenkel-Weiner \cite{Frenkel2015}, Bob does not receive any input $y$ therein. That is why Bob's output $\kappa$ depends only on the message $m$ sent by Alice and classical shared randomness $\tilde{\lambda}$. 

Let us now focus on the set of probabilities $\mathcal{C}_d$ wherein $(\lambda,\tilde{\lambda})$ is the pre-shared randomness. Take into account the following decoding function,
\be \label{newde}
p_b(b_y|y,m,\lambda,\tilde{\lambda}) = \sum_{\kappa} P_{y,\lambda}(b_y|\kappa) p_b(\kappa|m,\tilde{\lambda}) ,
\ee 
where $\{P_{y,\lambda}(b_y|\kappa)\}$ is the post-processing defined in \eqref{comMlambda} and $p_b(\kappa|m,\tilde{\lambda})$ is given in \eqref{rxgk}. One can check that this is indeed a valid decoding function.

Next, we show that an arbitrary $p(b_y|x,y) \in \mathcal{Q}^C_d$ can always be reproduced by a suitable classical strategy \eqref{pc} involving pre-shared randomness $(\lambda,\tilde{\lambda})$ and the decoding function \eqref{newde}. 
With the help of \eqref{comMlambda}, \eqref{rxgk}, and \eqref{newde} in a subsequent order, we find
\bea
&& \int_\lambda \pi(\lambda) \tr \left(\rho_{x,\lambda} M_{b_y|y,\lambda}\right) \ d\lambda \nonumber \\
 & =& \sum_{\kappa} \int_\lambda    P_{y,\lambda}(b_y|\kappa) \pi(\lambda) \tr (\rho_{x,\lambda} G_{\kappa} ) \ d\lambda  \nonumber \\
& =& \sum_{m} \int_{\lambda} \int_{\tilde{\lambda}} \pi(\lambda) \tilde{\pi}(\tilde{\lambda}) p_a(m|x,\lambda,\tilde{\lambda}) \nonumber \\
 && \qquad \qquad \times \left( \sum_{\kappa}  P_{y,\lambda}(b_y|\kappa) p_b(\kappa|m,\tilde{\lambda}) \right) d\lambda \   d\tilde{\lambda} 
\nonumber \\
&=& \sum_{m} \int_\lambda \int_{\tilde{\lambda}} \pi(\lambda) \tilde{\pi}(\tilde{\lambda}) p_a(m|x,\lambda,\tilde{\lambda}) p_b(b_y|y,m,\lambda,\tilde{\lambda}) d\lambda   d\tilde{\lambda}. \nonumber \\
\eea 
Therefore, an arbitrary probability distribution $p(b_y|x,y)$ obtainable from a compatible set of measurements can be reproduced by a suitable classical strategy, implying that $\mathcal{Q}^C_d \subseteq \mathcal{C}_d$. 
On the other hand, any classical strategy is always realized by quantum strategy with compatible measurements, i.e., $\mathcal{C}_d \subseteq  \mathcal{Q}^C_d $. Therefore, these two sets are identical. 

Finally, we remark that the figure of merit of any communication task is some arbitrary function of the probabilities $p(b_y|x,y)$, we can infer that any quantum advantage (with or without shared randomness) in such tasks over classical communication can be attained only if the set of measurements is incompatible. This completes the proof.
\end{proof} 
Note that a weaker version of \textit{Result 1}, $\mathcal{\overline{Q}}_d^C \subseteq \mathcal{C}_d $, has been proposed in Ref. \cite{carlos}. However, from the arguments presented therein, it is not clear how their Eq.(4) follows, which raises questions upon the status of the last paragraph of the proof of their Theorem 1 presented in the
Appendix A of \cite{carlos}.

One profound implication of \textit{Result} \ref{result:QcC} is that any communication task can serve as a witness of measurement incompatibility. 
However, in general, \textit{measurement incompatibility is not sufficient for quantum advantage without pre-shared randomness} \cite{PRXQuantum2021} -- there exist incompatible qubit measurements such that the set of probabilities given by \eqref{pqb} for arbitrary quantum states is within $\mathcal{C}_2$ (see Sec. IV-A of \cite{PRXQuantum2021}).

Another useful implication of \textit{Result} \ref{result:QcC} appears when we are interested in linear functions of $\{p(b_y|x,y)\}$. Communication tasks for which the figures of merit are linear functions of $\{p(b_y|x,y)\}$ are widespread due to their practical importance in quantum communication complexity tasks \cite{RevModPhys.82.665,DEWOLF,saha2019}, quantum key distribution \cite{Marcin2011}, quantum randomness generation \cite{randomness,randomness1}, quantum random access codes \cite{brac,Carmeli2020}, oblivious transfer \cite{Spekkens2009,saha2019}, and many other applications.
\begin{fact}\label{result2}
Given any scenario specified by $(l,n,\Vec{d} )$, the maximum values of any linear function of $\{p(b_y|x,y)\}$ obtained within the three different sets $\mathcal{C}_d$, $\mathcal{Q}^C_d,$ and $\mathcal{\overline{Q}}^C_d$ are the same.
\end{fact} 
\begin{proof}
To find the optimum value of any linear function of $\{p(b_y|x,y)\}$, it is sufficient to consider classical strategy without shared randomness (see \textit{Lemma} 1 in Appendix \ref{app2} for a detailed explanation). Now,  all probability distributions $\{p(b_y|x,y)\}$, which are obtained from classical strategy without shared randomness, can always be realized within $\mathcal{\overline{Q}}^C_d$. Upon receiving the input $x$, Alice sends the quantum state $\rho_x$ such that $\rho_x$ is diagonal in the computational basis. Bob, upon receiving the input $y$ and the diagonal state $\rho_x$, performs a fixed measurement $\{G_\kappa \}$ in the computational basis followed by some post-processing depending on $y$. This observation together with the relation $\mathcal{\overline{Q}}_d^C \subseteq \mathcal{Q}^C_d = \mathcal{C}_d$ implies the above fact. 
\end{proof}

The generic relation among the sets $\mathcal{\overline{Q}}_d^C, \mathcal{Q}^C_d, \mathcal{\overline{Q}}_d, \mathcal{Q}_d, \mathcal{C}_d$ is further analyzed in Section \ref{sec5}. Next, we will propose incompatibility witness for an arbitrary set of measurements using a family of communication tasks, namely, the general version of RAC \cite{racambainis}.

	\begin{figure}[t!]
		\centering
		\includegraphics[scale=0.43]{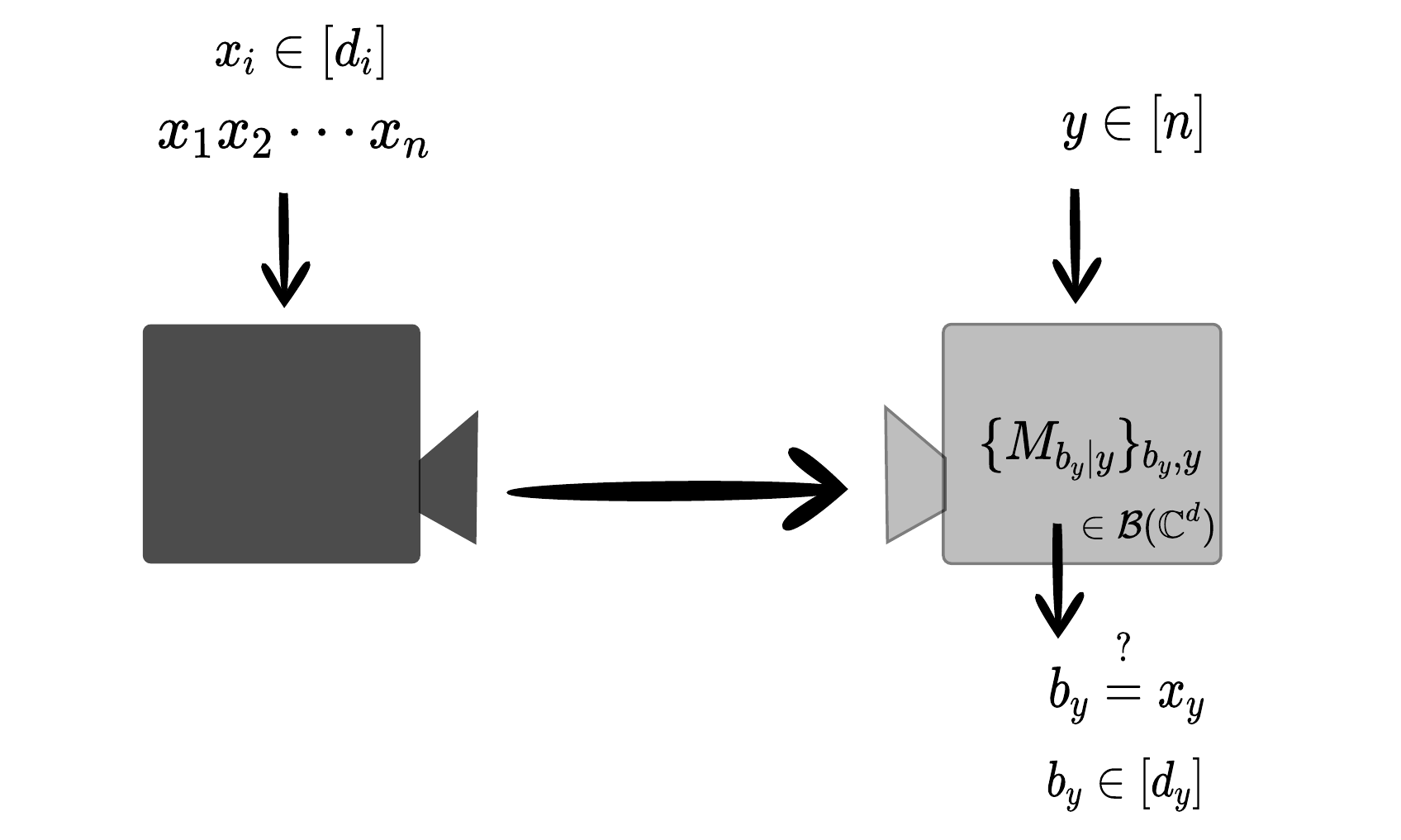}
		\caption{An unknown measurement set of arbitrary settings, $\{M_{b_{y}|y}\}_{b_y,y}$, is provided; we only know the dimension ($d$) on which this set of measurements act. Our task is to certify the incompatibility of this set of measurements. Here, the notation $[k]:= \{1,\cdots,k\}$ for any natural number $k$.}\label{fig1}
	\end{figure}

\section{Incompatibility witness for any set of measurements of arbitrary setting} \label{sec4}

Take the most general form of a set of measurements. There are $n$ measurements, defined by $\{M_{b_y|y}\}$ where $y\in [n]$ each of which has different outcomes, say, measurement $y$ has $d_y$ outcomes, that is, $b_y \in [d_y]$. These measurements are acting on $d$-dimensional quantum states where $d$ is finite (see FIG. \ref{fig1}). In order to witness incompatibility of this set, we introduce the most general form of random access codes with Bob having this set of measurements. Alice gets a string of $n$ dits $x=x_1x_2\cdots x_n$ randomly from  the set of all possible strings in which $x_y \in [d_y]$ for all $y \in [n]$. While Alice communicates a $d$-dimensional classical or quantum system to encode the information about the obtained string, the task for Bob is to guess the $y$-th dit when $y$ is chosen randomly.  The figure of merit is the average success probability defined by the following linear function
\be \label{Pndd}
S(n,\Vec{d},d) = \frac{1}{n\prod_y d_y} \sum_{x,y} p(b_y=x_y|x,y)
\ee 
which is fully specified by $n$, $\Vec{d}=(d_1,d_2,\cdots ,d_n)$, and $d$. Since Eq.\eqref{Pndd} is a linear function of $p(b_y|x,y)$, by Fact \ref{result2}, the maximum values over $\mathcal{C}_d$, $\mathcal{Q}^C_d$ and $\mathcal{\overline{Q}}^C_d$ are the same and denoted by $S_c(n,\Vec{d},d)$. Precisely,
\begin{align}
   S_c(n,\Vec{d},d) &= \max_{\{p(b_y|x,y)\} \in \mathcal{C}_d} S(n,\Vec{d},d) \nonumber \\
   &= \max_{\{p(b_y|x,y)\} \in \mathcal{Q}^C_d} S(n,\Vec{d},d).
\end{align}
Hence, $S_c(n,\Vec{d},d)$ can be evaluated by maximizing the average success probability either over all classical strategies, or over all quantum strategies involving compatible measurements only. 

Whenever a set of measurements in the scenario specified by $n$, $\Vec{d}$, $d$ gives $S(n,\Vec{d},d) > S_c(n,\Vec{d},d)$ in the above-introduced general version of RACs, we can conclude that the measurements are incompatible. Hence, in order to witness measurement incompatibility, we need to know $S_c(n,\Vec{d},d)$.
Now we present an upper bound on $S_c(n,\Vec{d},d)$ for arbitrary $n,\Vec{d},d$.
\begin{result}\label{thm:gracb}
The following relation holds true for arbitrary $(n,\Vec{d},d):$
\be \label{SC1}
S_c(n,\Vec{d},d) \leqslant \frac{1}{n}  \times \min  \Bigg\{ 1 + \sum_{\substack{i,j\\i<j}} \frac{d}{ d_id_j } \ , \ n-1+\frac{d}{\prod_y d_y} \Bigg\} .
\ee
\end{result}
This upper bound in Eq. \eqref{SC1} is obtained for $\mathcal{Q}^C_d$, that is, by taking the existence of a parent POVM of the measurements $\{M_{b_y|y}\}_{b_y,y}$ performed by Bob. The proof of this result is presented in Appendix \ref{app1}.  
When the outcome of all the measurements is the same, which is $d_y = \od $ for all $y$, the above bound simplifies to
\be 
S_c(n,\od,d) \leqslant \frac{1}{n} \times \min \left\{1 + \frac{n(n-1)d}{2\od^{2}} \ , \ n-1 +\frac{d}{\od^n} \right\} .
\ee 
Hence, in different types of RACs involving an arbitrary set of quantum measurements by Bob, if the average success probability exceeds the aforementioned upper bounds on $S_c$, then we can conclude that the measurements by Bob are incompatible. 

On the other hand, whenever
\be \label{ddy}
d \leqslant \min_y d_y
\ee 
we find out the exact value of $S_c(n,\Vec{d},d)$. 
Say, $k_i$ is the number of sets among $[d_1],\cdots,[d_n]$ such that dit $i\in [d_y]$. For example, consider the RACs with $n=4$ and $d_1=2$, $d_2=3$, $d_3=4$ and $d_4 = 3$. That is, Alice gets a string of four dits $x=x_1 x_2 x_3 x_4$ randomly, where $x_1 \in [2]$, $x_2 \in [3]$, $x_3 \in [4]$ and $x_4 \in [3]$. In this case, $k_1=4$, $k_2=4$, $k_3=3$, and $k_4=1$.
Also, we denote $d_{\max} = \max_y d_y $.

\begin{result}\label{thm:graceb}
If \eqref{ddy} holds, then
\be \label{SCex}
S_c(n,\Vec{d},d) = \frac{1}{n\prod_y d_y} \sum \left[ \left(\prod_{j=1}^{d_{\max}}  C^{\alpha_j}_{n_j} \right)\, \max_{i=1,\cdots,d}\{n_i\} \right]
\ee 
with 
\begin{align}
    \alpha_j = k_j - \sum_{i=j+1}^{d_{\max}} n_i, \quad   C^{\alpha_j}_{n_j} = \frac{\alpha_j (\alpha_j-1) \cdots (\alpha_j - n_j+1) }{n_j (n_j-1) \cdots  1}  \nonumber
\end{align}
and where the summation is taken over all possible integer solutions of the following equation 
\be 
\sum_{i=1}^{d_{\max}} n_i = n
\ee 
such that $n_i \leqslant k_i$ for all $i$. 
\end{result}
Note here that (\ref{SCex}) is obtained for $\mathcal{C}_d$ by considering  classical strategies. The detailed proof is given in Appendix \ref{app2}. For a particular case of {\it Result} \ref{thm:graceb} wherein $d_y = \od = d$ for all $y$, the proof is previously given in \cite{brac}.
Hence, when $d \leqslant \min_y d_y$,  a necessary criteria for a set of measurements to be compatible is given by
\begin{align} \label{witness-ineq}
    S(n,\Vec{d},d) \leqslant S_c(n,\Vec{d},d),
\end{align}
where $S_c(n,\Vec{d},d)$ is given by (\ref{SCex}).

For $n=2$, $d_y=\od$ for all $y$, and $d\leqslant \od$, the expression \eqref{SCex} simplifies to (for details, see Appendix \ref{app3})
\be \label{simSCex}
S_c(2,\od,d) = \frac{1}{2\od^2} \left(d + 2d\od -d^2 \right) .
\ee 
And for $n=3$, $d_y=\od$ for all $y$, and $d\leqslant \od$, the expression \eqref{SCex} simplifies to (for details, see Appendix \ref{app3})
\be \label{3_measurements}
S_c(3,\od,d) = \frac{d}{3\od^3} \left(d^2 - 1 + 3\od (\od + 1 - d) \right) .
\ee 
The particular case of \textit{Result} \ref{thm:gracb} for $n=2$ can be found in \cite{Carmeli2020}, and, moreover, it is shown that any pair of rank-one projective measurements that is incompatible provides advantages in RAC \cite{saha2020}. 
In order to showcase the generic applicability of {\it Results} \ref{thm:gracb} and \ref{thm:graceb}, we consider an arbitrary set of three rank-one projective qubit measurements, which using some unitary can be expressed as 
\bea \label{3qpm}
    M_{x_1|1} &=& (1/2) U \left[\I + (-1)^{x_1} \sigma_z \right] U^\dagger \nonumber \\
    M_{x_2|2} &=& (1/2) U \left[\I + (-1)^{x_2} \left( \alpha \sigma_z + \sqrt{1-\alpha^2} \sigma_x \right) \right] U^\dagger  \nonumber \\
    M_{x_3|3} &=& (1/2) U \Big[\I + (-1)^{x_3} \big( \beta \sigma_z + \gamma \sqrt{1-\beta^2}  \sigma_x \nonumber \\ 
    && \qquad \qquad \quad \pm \sqrt{1-\beta^2} \sqrt{1-\gamma^2}  \sigma_y \big) \Big] U^\dagger  
\eea 
where $x_1,x_2,x_3 \in [2],$ the variables $\alpha,\beta,\gamma \in [-1,1]$, and $U$ can be an arbitrary unitary operator acting on $\mathbbm{C}^2$. We obtain the following result.
\begin{result}\label{thm:322}
The figure of merit \eqref{Pndd} of RACs for $n=3,\od =2,d=2$ can witness  any set of three incompatible rank-one projective qubit measurements, except for the sets defined by (\ref{3qpm}) with
\be 
(\alpha,\beta,\gamma ) = \{ (\pm 1/2,\pm 1/2,-1),(\pm 1/2,\mp 1/2,1)\}.  \nonumber
\ee  
\end{result} 
This result is proved with the help of numerical optimizations, and the proof is 
available in Appendix \ref{app:322}. 


\section{Generic relations between probability sets}\label{sec5}


We now point out a few generic relations between the sets in order to get a generic perspective. It is trivial that $\mathcal{\overline{Q}}_d \not\subset \mathcal{C}_d$ since we observe the quantum advantage for $S(n,\Vec{d},d)$. We also observe the following,
\begin{fact}\label{fact:2}
In general, $\mathcal{C}_d \not\subset \mathcal{\overline{Q}}_d$, and thus, $\mathcal{\overline{Q}}^C_d \subsetneq \mathcal{C}_d$. In other words, there exist probabilities that belong to $\mathcal{C}_d$ but do not belong to $\mathcal{\overline{Q}}_d$. 
Moreover, in general, $\mathcal{Q}_d \setminus (\mathcal{C}_d \cup \mathcal{\overline{Q}}_d ) \neq \varnothing$. In other words, there exists a probability distribution that belongs to $\mathcal{Q}_d$ but does not belong to $\mathcal{C}_d$ and  $\mathcal{\overline{Q}}_d$.
\end{fact}
\begin{proof} Once again, reckon the RAC task and, instead of the average success probability \eqref{Pndd}, let us consider the figure of merit to be the worst case success probability,
\be 
W(n,\tilde{d},d) = \min_{x,y} \{p(b_y=x_y|x,y) \}.
\ee 
It follows from Yao's principle that the average success probability $S(n,2,2)$ is the same as $W(n,2,2)$ when pre-shared randomness is available. See Lemma 1 in \cite{ambainis2009quantum} for the proof. 
Therefore, $W(4,2,2)$ in $\mathcal{C}_2$ is the same as $S_c(4,2,2) = 11/16$ \cite{ambainis2009quantum}. However, it was proven that the best value of $W(4,2,2)$ in $\mathcal{\overline{Q}}_2$ is 1/2 \cite{Hayashi_2006}, implying $\mathcal{C}_d \not\subset \mathcal{\overline{Q}}_d$. Moreover, since $\mathcal{\overline{Q}}^C_d$ is a subset of both $\mathcal{C}_d$ and $\mathcal{\overline{Q}}_d$, it must be a proper subset.

Besides, there exists a quantum strategy in $\mathcal{Q}_2$ that achieves $W(4,2,2) = 0.74$ (section 4.1.2 in \cite{ambainis2009quantum}). This means that both $\mathcal{C}_d$ and $\mathcal{\overline{Q}}_d$ are proper subsets of $\mathcal{Q}_d$. 
\end{proof}

Next, we present another observation that helps in completely understanding the relationship between different sets of probability distributions.

\begin{fact}\label{fact:w322}
In general, $(\mathcal{C}_d \cap \mathcal{\overline{Q}}_d ) \setminus \mathcal{\overline{Q}}^C_d \neq \varnothing $. In other words, there exists a set of probabilities that does not belong to $\mathcal{\overline{Q}}^C_d$ at the same time inside both of the sets $\mathcal{C}_d$ and $\mathcal{\overline{Q}}_d$.
\end{fact}
\begin{proof} In $\mathcal{C}_d$, $W(3,2,2)=S_c(3,2,2)=3/4$ \cite{ambainis2009quantum}. It is also well known that there exist two-dimensional quantum states and measurements that lead to $W(3,2,2) = 1/2 + 1/(2\sqrt{3}) \approx 0.79$, without pre-shared randomness. An example of such quantum states is given in \textit{Example 2} of Ref. \cite{Hayashi_2006}. If we consider the noisy version of those quantum states $\nu \rho + (1-\nu) \I/2$ such that $\nu = \sqrt{3}/2$, then it is easy to see that the value of $W(3,2,2)$ reduces to 3/4. Therefore, $W(3,2,2)=3/4$ is obtained in both of the sets $\mathcal{C}_2$ and $\mathcal{\overline{Q}}_2$. On the other hand, a numerical analysis using semi-definite programming shows that $W(3,2,2)$ is, at most, 2/3 in $\mathcal{\overline{Q}}^C_2$, which implies $(\mathcal{C}_2 \cap \mathcal{\overline{Q}}_2 ) \setminus \mathcal{\overline{Q}}^C_2 \neq \varnothing $.
\end{proof}

Based on the above observations, the generic relationship between different sets of probabilities can be understood as depicted in FIG. \ref{fig2}.

\begin{figure}[t!]
		\centering
		\includegraphics[scale=0.29]{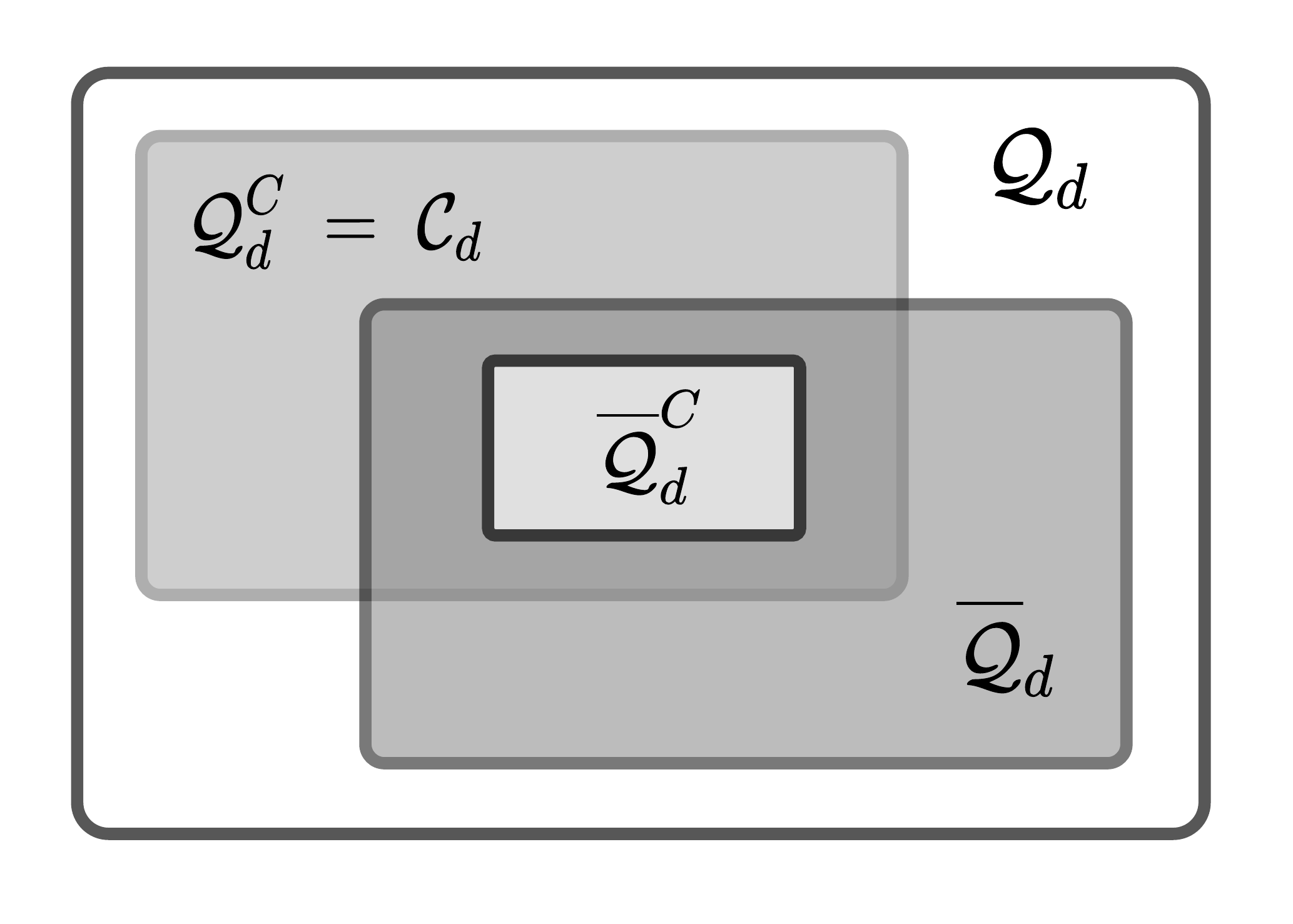}
		\caption{Generic relation between different sets of probabilities. The equivalence between $\mathcal{Q}_d^C$ and $\mathcal{C}_d$ is shown in \textit{Result 1}. There exist many examples such that $\mathcal{\overline{Q}}_d$ is not a subset of $\mathcal{C}_d$. On the other hand, Fact \ref{fact:2} points out that $\mathcal{C}_d$ is not a subset of $\mathcal{\overline{Q}}_d$ and both $\mathcal{C}_d$ as well as $\mathcal{\overline{Q}}_d$ are proper subsets of $\mathcal{Q}_d$. While it is immediate that $\mathcal{\overline{Q}}^C_d$ is a subset of both $\mathcal{C}_d$ and $\mathcal{\overline{Q}}_d$, Fact \ref{fact:w322} shows that $\mathcal{\overline{Q}}^C_d$ is not the intersection of these two sets.}\label{fig2}
\end{figure}

\section{Concluding remarks}\label{sec6}

By characterizing the set of quantum correlations in prepare-and-measure scenarios produced from any set of compatible measurements, we have shown in this article that  incompatible measurements at the receiver's end are necessary to demonstrate a quantum advantage in any communication task. Based on this result, we have presented a semi-device-independent witness of measurement incompatibility invoking generalized random access codes. Further, we have completely characterized the sets of three incompatible projective qubit measurements that can be detected using our proposed witness. The relationship between the classical probability distributions and different types of quantum probability distributions produced in an arbitrary communication task has also been presented.

It might be noted that some of the results derived in \cite{Carmeli2020,saha2020} appear as natural corollaries of the results obtained here. It has recently been shown that the measurement statistics produced in any communication task involving compatible measurements by the receiver can  be reproduced by classical communication, where the dimension of the classical communication is of the order of the number of outcomes of the parent POVM \cite{ioannou}. Significantly, our results further reduce  the dimension of the classical communication involved.

The importance of the analysis presented here lies in the fact that the classical bound of the  success metric of any communication task becomes an upper bound on the metric of the task under a compatible set of measurements. Consequently, violating the classical bound of any communication task can be used as a sufficient criterion to witness measurement incompatibility.
These bounds are tight whenever the dimension on which the measurements act is not larger than the number of outcomes of any of the measurements. The present study establishes that measurement incompatibility is the fundamental quantum resource for non-classicality in any communication task or, more generally, in prepare-and-measure scenarios.

Our analysis paves the way for investigations of several open
questions. First, deriving more efficient incompatibility witnesses based on different communication tasks is worthwhile for future studies. Second, our results may be generalized to propose  semi-device-independent witnesses for incompatible quantum channels \cite{doi:10.1063/1.5126496} and quantum instruments \cite{IncomQI,PhysRevA.105.052202}. 
Finally, proposing operational witnesses for all incompatible extremal POVMs \cite{Sent_s_2013} is another fundamentally motivated open problem.  \\

\begin{acknowledgements}
DS and DD acknowledge the Science and Engineering Research Board (SERB), Government of India for financial support through National Post Doctoral Fellowship (PDF/2020/001682 and PDF/2020/001358, respectively). During the later phase of this work, the research of DD has been supported by the Royal Society (United Kingdom) through the Newton International Fellowship (NIF$\backslash$R$1\backslash212007$). BB acknowledges the INSPIRE fellowship from the Department of Science and Technology, Government of India. AKD and ASM acknowledge support from
Project No. DST/ICPS/QuEST/2018/98 from the Department of Science  and
Technology, Govt. of India.
\end{acknowledgements}

\bibliography{ref} 

\onecolumngrid  


\appendix 

\section{Proof of Result \ref{thm:gracb}} \label{app1}

\begin{proof}
For a given set of measurements $\{M_{b_y|y}\}$ on Bob's side, the maximum average success probability \eqref{Pndd} in quantum theory is given by \cite{saha2020}
\bea  \label{maxev}
&& \max_{\{\rho_x\}} \ \frac{1}{n \prod_y d_y}  \ \sum_{ x_{1}x_{2}\cdots x_{n}} \sum_y \tr\left(\rho_x M_{b_y=x_y|y} \right) \nonumber \\
&=& \frac{1}{n \prod_y d_y}  \ \sum_{ x_{1}x_{2}\cdots x_{n}}  \max_{\{\rho_x\}}  \tr\left(\rho_x \left(\sum_y M_{b_y=x_y|y} \right)\right) \nonumber \\
&=& \frac{1}{n \prod_y d_y} \ \sum_{ x_{1}x_{2}\cdots x_{n}} ||\chi||,
\eea 
where, 
\be 
\chi = M_{x_{1}|1} + M_{x_{2}|2}+ \cdots + M_{x_{n}|n} .
\ee 
Here, $||\chi||$ denotes the operator norm of $\chi$, which is simply the maximum eigenvalue of the operator $\chi$. Our aim is to obtain an upper bound on the expression \eqref{maxev} when $\{M_{b_y|y}\}$ are compatible. 
An alternative definition of measurement incompatibility, which is equivalent to the standard one \eqref{com}, is associated with the existence of a parent POVM whose appropriate marginals give rise to all the individual measurements \cite{review}.
Precisely, if the measurements $\{M_{b_y|y}\}$ are compatible then there exists a parent measurement,  $G\equiv \{G(b_{1},\cdots,b_{n})\}$, with $\prod_{y=1}^{n} d_y$ elements from which all the measurement operators can be reconstructed by taking marginals as follows
\be \label{marginal}
M_{x_{y}|y} =  \sum_{b_1, \cdots, b_{y-1}, b_{y+1}, \cdots, b_n} G(b_{1}, \cdots , b_{y_1}, x_y, b_{y+1}, \cdots, b_{n}),
\ee
where
\be \label{sumG}
\sum_{b_1, \cdots, b_n} G(b_{1}, \cdots , b_{y_1}, b_y, b_{y+1}, \cdots, b_{n}) =\I_{\small{d\times d}}. 
\ee 
Let us first expand $\chi$ in terms of the parent POVM using \eqref{marginal},
 \be \label{chi}
\chi = \sum_{b_{2},b_3,\cdots,b_{n}}G(x_{1},b_{2},b_3\cdots,b_{n})+\sum_{b_{1},b_{3},\cdots,b_{n}}G(b_{1},x_{2},b_{3},\cdots,b_{n})+ \cdots +
\sum_{b_{1},b_{2},b_{3},\cdots,b_{n-1}}G(b_{1},b_{2},b_{3},\cdots,b_{n-1},x_{n}).  
\ee 
Each term in the above expansion can be split into two terms in the following way
\bea \label{chi2}
\chi &=& \sum_{b_{3},b_{4},\cdots,b_{n}}G(x_{1},x_{2},b_{3}, b_4, \cdots,b_{n})+\sum_{\substack{b_{2},\cdots,b_{n}\\ b_{2}\neq x_{2}}}G(x_{1},b_{2},\cdots,b_{n}) \nonumber \\
&& + \sum_{b_{1},b_{4},\cdots,b_{n}}G(b_{1},x_{2},x_{3},b_{4},\cdots,b_{n})+\sum_{\substack{b_{1},b_{3},\cdots,b_{n}\\ b_{3}\neq x_{3}}}G(b_{1},x_{2},b_{3},\cdots,b_{n}) \nonumber \\
&& +\cdots \nonumber \\
&& + \sum_{b_{2},b_{3},\cdots,b_{n-1}}G(x_{1},b_{2},b_{3},\cdots,b_{n-1},x_{n}) + \sum_{\substack{b_{1},\cdots,b_{n-1}\\ b_{1}\neq x_{1}}}G(b_{1},\cdots,b_{n-1},x_{n}).
\eea
In the above Eq. \eqref{chi2}, there are two sums in each line, and there is a total of $n$ lines. Let us denote the first sum and the second sum in the $ i^{\text{th}}  $ line by $\mathcal{S}_1^i$ and $\mathcal{S}_2^i$ respectively, where $i \in \{1, \cdots, n\}$. Hence, Eq. \eqref{chi2} can be expressed as 
\begin{align} \label{chiS1S2}
    \chi = \sum_{i=1}^{n} \left( \mathcal{S}_1^i + \mathcal{S}_2^i \right),
\end{align}
where
\begin{align}
    \mathcal{S}_1^i = \sum_{b_1, \cdots, b_{i-1}, b_{i+2}, \cdots, b_n} G(b_1, \cdots, b_{i-1}, x_i, x_{i+1}, b_{i+2}, \cdots, b_n),
\end{align}
and 
\be 
\mathcal{S}_2^i = \sum_{\substack{b_{1},\cdots,b_{i-1},b_{i+1},\cdots,b_{n}\\ b_{i+1}\neq x_{i+1}}}G(b_{1},\cdots,b_{i-1},x_{i},b_{i+1},\cdots,b_{n}). 
\ee   
Here, the index $i$ is taken to be modulo $n$. Each $G(\cdots)$ in the above sums will be termed as an element. 

Let us now make an observation that there is no common element between the $\mathcal{S}_2^i$ and $\mathcal{S}_2^{i+1}$. The common element between $\mathcal{S}_2^i$ and $\mathcal{S}_2^j$ with $i, j \in \{1, \cdots, n\}$ and $j>i+1$  is
\begin{align}
&\sum_{\substack{b_1,\cdots , b_n\\ b_{i+1}\neq x_{i+1}\\b_{j+1}\neq x_{j+1}}} G(b_{1},\cdots,b_{i-1},x_{i},b_{i+1},\cdots,b_{j-1},x_{j},b_{j+1},\cdots,b_{n}) \nonumber \\
&\leqslant \sum^n_{\substack{k=1\\k\neq i, j}} \sum_{b_k} G(b_{1},\cdots,b_{i-1},x_{i},b_{i+1},\cdots,b_{j-1},x_{j},b_{j+1},\cdots,b_{n})
\end{align}
where the index $i,j$ is taken to be modulo $n$. 
Hence, we have
\be \label{S2i}
\sum_{i=1}^n \mathcal{S}_2^i \leqslant   \sum_{\substack{i,j \in \{1, \cdots, n\} \\j>i+1}} \bigg( \sum^n_{\substack{k=1\\k\neq i, j}} \sum_{b_k} G(b_{1},\cdots,b_{i-1},x_{i},b_{i+1},\cdots,b_{j-1},x_{j},b_{j+1},\cdots,b_{n}) \bigg) + \text{other terms with no common element}.
\ee
Next, let us focus on $\mathcal{S}_1^i$. It can be checked that
\be \label{S1i}
\sum_{i=1}^{n} \mathcal{S}_1^i = \sum_{i=1}^{n} \bigg(\sum^n_{\substack{k=1\\k\neq i, i+1}} \sum_{b_k} G(b_{1},\cdots,b_{i-1},x_{i},x_{i+1},b_{i+2},\cdots,b_{n}) \bigg).
\ee 
Replacing $\mathcal{S}_2^i$ and $\mathcal{S}_1^i$ using \eqref{S2i} and \eqref{S1i} in \eqref{chiS1S2}, we have 
\bea 
\chi &\leqslant &  \sum_{\substack{i,j \in \{1, \cdots, n\}\\i<j}} \bigg( \sum^n_{\substack{k=1 \\ k\neq i,j}} \sum_{b_k} G(b_{1},\cdots,b_{i-1},x_{i},b_{i+1},\cdots,b_{j-1},x_{j},b_{j+1},\cdots,b_{n}) \bigg) + \text{other terms with no common element} \nonumber \\
& \leqslant & \sum_{\substack{i,j \in \{1, \cdots, n\}\\i<j}} \bigg( \sum^n_{\substack{k=1 \\ k\neq i,j}} \sum_{b_k} G(b_{1},\cdots,b_{i-1},x_{i},b_{i+1},\cdots,b_{j-1},x_{j},b_{j+1},\cdots,b_{n}) \bigg) +  \sum_{b_1,\cdots , b_n}G(b_{1},\cdots,b_{n}) .
\eea  
Substituting the above expression into \eqref{maxev} and employing the triangle inequality for the norm, we find that
\bea \label{SCexpand}
S_c(n,\Vec{d},d) &\leqslant & \frac{1}{n \prod_y d_y} \Bigg( \sum_{x_{1},\cdots,x_{n}} \, \, \, \,  \sum_{\substack{i,j \in \{1, \cdots, n\}\\i<j}} \, \,  \big{|}\big{|} \sum^n_{\substack{k=1 \\ k\neq i,j}} \sum_{b_k} G(b_{1},\cdots,b_{i-1},x_{i},b_{i+1},\cdots,b_{j-1},x_{j},b_{j+1},\cdots,b_{n}) \big{|}\big{|} \nonumber \\
 && \qquad \qquad \quad + \sum_{x_{1},\cdots,x_{n}} \big{|}\big{|} \sum_{b_1,\cdots , b_n} G(b_{1}, \cdots,b_{n})\big{|}\big{|} \Bigg).
 \label{bigsum}
\eea
Due to \eqref{sumG}, the second term of the above expression can be evaluated as 
\be \label{2ndterm}
\sum_{x_{1},\cdots,x_{n}}\big{|}\big{|} \sum_{b_1,\cdots , b_n}G(b_{1}, \cdots,b_{n})\big{|}\big{|}  = \sum_{x_{1},\cdots,x_{n}}\big{|}\big{|} \I_{d\times d} \big{|}\big{|} = \prod_{y=1}^n d_y .
\ee 
%
%
Next, consider the first term in (\ref{bigsum}), given by
\begin{align}
  &\sum_{x_{1},\cdots,x_{n}} \, \, \, \,  \sum_{\substack{i,j \in \{1, \cdots, n\}\\i<j}} \, \,  \big{|}\big{|} \sum^n_{\substack{k=1 \\ k\neq i,j}} \sum_{b_k} G(b_{1},\cdots,b_{i-1},x_{i},b_{i+1},\cdots,b_{j-1},x_{j},b_{j+1},\cdots,b_{n}) \big{|}\big{|}  =   \sum_{\substack{i,j \in \{1, \cdots, n\}\\i<j}} \, \, \beta_{i,j} 
\end{align}
where
\bea \label{1stterm}
\beta_{i,j} & = & \sum_{x_{1}, \cdots,x_{n}} \big{|}\big{|} \sum^n_{\substack{k=1 \\ k\neq i,j}} \sum_{b_k} G(b_{1},\cdots,b_{i-1},x_{i},b_{i+1},\cdots,b_{j-1},x_{j},b_{j+1},\cdots,b_{n})  \big{|}\big{|} \nonumber \\
& \leqslant & \sum_{x_{1},\cdots,x_{n}} \text{Tr} \left( \sum^n_{\substack{k=1 \\ k\neq i,j}} \sum_{b_k} G(b_{1},\cdots,b_{i-1},x_{i},b_{i+1},\cdots,b_{j-1},x_{j},b_{j+1},\cdots,b_{n})  \right) \nonumber \\
&=& \sum^n_{\substack{r = 1\\ r\neq i,j}} \sum_{x_r} \text{Tr} \left( \sum_{x_{i},x_{j}} \sum^n_{\substack{k=1 \\ k\neq i,j}} \sum_{b_k} G(b_{1},\cdots,b_{i-1},x_{i},b_{i+1},\cdots,b_{j-1},x_{j},b_{j+1},\cdots,b_{n})  \right) \nonumber \\ 
&=& \sum^n_{\substack{r=1 \\ r\neq i,j}} \sum_{x_r} \text{Tr} \left(  \sum_{b_1, \cdots, b_n} G(b_{1},\cdots,b_{n})  \right) \nonumber \\ 
& = & \sum^n_{\substack{r=1 \\ r\neq i,j}} \sum_{x_r}   \text{Tr} (\I_{d\times d}) \nonumber \\
&=& d \prod^n_{\substack{y=1 \\ y\neq i,j}} d_y .
\eea
Hence, we get the first term in (\ref{bigsum}), given by
\begin{align}
&\sum_{x_{1},\cdots,x_{n}} \, \, \, \,  \sum_{\substack{i,j \in \{1, \cdots, n\}\\i<j}} \, \,  \big{|}\big{|} \sum^n_{\substack{k=1 \\ k\neq i,j}} \sum_{b_k} G(b_{1},\cdots,b_{i-1},x_{i},b_{i+1},\cdots,b_{j-1},x_{j},b_{j+1},\cdots,b_{n}) \big{|}\big{|} =  d  \sum_{\substack{i,j \in \{1, \cdots, n\}\\i<j}} \, \, \prod^n_{\substack{y=1 \\ y\neq i,j}} d_y 
\label{firstterm}
\end{align}

By substituting the bounds from \eqref{2ndterm}-\eqref{firstterm} into \eqref{SCexpand}, we obtain
\be
S_c(n,\Vec{d},d) \leqslant \frac{1}{n \prod_y d_y} \left( d \sum_{\substack{i,j \in \{1, \cdots, n\} \\i<j}} \left(\prod^n_{\substack{y=1 \\ y\neq i,j}} d_y \right) + \prod^n_{y=1} d_y \right) ,
\ee
which reduces to the first expression of Eq. \eqref{SC1}.

For the other bound, let us use the fact that in \eqref{chi}, only the term $G(x_1,x_2,\cdots,x_n)$ occurs $n$ times and all the other terms can occur, at most, $(n-1)$ times to get an upper bound on $\chi$ as follows
\be 
\chi \leqslant G(x_1,x_2,\cdots,x_n) + (n-1) \sum_{b_1,\cdots,b_n} G(b_1,\cdots,b_n) .
\ee 
Replacing this bound into \eqref{sumG} and employing the triangle inequality for the norm, we get
\be 
S_c(n,\Vec{d},d) \leqslant \frac{1}{n \prod_y d_y} \left( \sum_{x_{1},\cdots,x_{n}} \big{|}\big{|} G(x_{1},\cdots,x_{n})\big{|}\big{|} + (n-1) \sum_{x_{1},\cdots,x_{n}} \big{|}\big{|} \sum_{b_1,\cdots , b_n}G(b_{1},\cdots,b_{n})\big{|}\big{|}   \right) . 
\ee 
We already have a bound given by \eqref{2ndterm} on the second sum in the above equation. The first term is bounded by $d$ since 
\begin{align}
\sum_{x_{1},\cdots,x_{n}} \big{|}\big{|} G(x_{1},\cdots,x_{n})\big{|}\big{|} &\leqslant \sum_{x_{1},\cdots,x_{n}} \text{Tr} \left( G(x_{1},\cdots,x_{n}) \right) \nonumber \\
& = \text{Tr}\left( \sum_{x_1, \cdots, x_n} G(x_{1},\cdots,x_{n}) \right) \nonumber \\
&=\text{Tr}\left( \I_{d\times d} \right) \nonumber \\
&= d .
\end{align} 
Therefore, we arrive at
\be 
S_c(n,\Vec{d},d) \leqslant \frac{1}{n \prod_y d_y} \left( d + (n-1) \prod_y d_y \right),
\ee 
which reduces to the second expression of Eq. \eqref{SC1}. 
This completes the proof.
\end{proof}

\section{Proof of Result \ref{thm:graceb}} \label{app2}

In order to provide a detailed proof of Result \ref{thm:graceb}, 
we first state a general feature of communication tasks.
\begin{lemma}
Consider a general form of a linear function of $\{p(b_y|x,y)\}$,
\be \label{S}
S = \sum_{x,y,b_y} c_{x,y,b_y} p(b_y|x,y) .
\ee 
The maximum value of $S$ within $\mathcal{C}_d$, which we denote by $S_c$, is obtained by deterministic strategies and can be written only in terms of decoding function $\{p_b(b_y|y,m)\}$.
\end{lemma}
\begin{proof}
Replacing the expression of $p(b_y|x,y)$ for classical communication given by Eq. \eqref{pc} 
into \eqref{S}, we see that
\begin{align} \label{simSCg}
S_c & = \max_{\substack{ \{p_a(m|x,\lambda)\}\\ \{p_b(b_y|y,m,\lambda)\}\\ \{\pi(\lambda)\}} }
\sum_x \bigg( \sum_{y,b_y} c_{x,y,b_y} p(b_y|x,y) \bigg) \nonumber \\
& =   \max_{\substack{ \{p_a(m|x,\lambda)\} \\ \{p_b(b_y|y,m,\lambda)\}\\ \{\pi(\lambda)\}} } \int_\lambda \pi(\lambda) \left[ \sum_{x} \left\{ \sum_m p_a(m|x,\lambda) \left(\sum_{y,b_y} c_{x,y,b_y} p_b(b_y|y,m,\lambda)\right) \right\} \right] d\lambda \nonumber \\
& =   \max_{\substack{ \{p_b(b_y|y,m,\lambda)\}\\ \{\pi(\lambda)\}} } \int_\lambda \pi(\lambda) \left[ \sum_{x} \max_m \left(\sum_{y,b_y} c_{x,y,b_y} p_b(b_y|y,m,\lambda)\right) \right] d\lambda . 
\end{align}
This is achieved when $p_a(m^*(\lambda,x) |x,\lambda) =1$, where for each $\lambda$, $m^*(\lambda,x)$ is defined as follows
\begin{align}
 \sum_{y,b_y} c_{x,y,b_y} p_b(b_y|y, m^*(\lambda,x),\lambda) \geq \sum_{y,b_y} c_{x,y,b_y} p_b(b_y|y,m,\lambda) \, \,  \, \forall \, \, m \in [d].
\end{align}
Now, the above expression (\ref{simSCg}) is a convex sum with respect to $\pi(\lambda)$ and thus, we can omit the dependence of $\lambda$ by taking the best value of $\left[\sum_{x} \max_m \left(\sum_{y,b_y} c_{x,y,b_y} p_b(b_y|y,m,\lambda)\right)\right]$ over different choices of $\lambda$ as follows:
\begin{align}
S_c =   \max_{ \{p_b(b_y|y,m)\} } \left[ \sum_{x} \max_m \left(\sum_{y,b_y} c_{x,y,b_y} p_b(b_y|y,m)\right) \right]. \label{simSCg3}
\end{align}
Therefore, it is sufficient to consider deterministic decoding, that is, $p_b(b_y|y,m) \in \{0,1\}$ to achieve $S_c$. Moreover, given any decoding strategy $\{ p_b(b_y|y,m) \}$, the best encoding function is 
\begin{align}
p_a(m^*|x) =1, \, \, \, \, \text{where} \, \, \,  \sum_{y,b_y} c_{x,y,b_y} p_b(b_y|y,m^*) \geq \sum_{y,b_y} c_{x,y,b_y} p_b(b_y|y,m) \, \,  \, \forall \, \, m \in [d].
 \label{besten}
\end{align}  
This completes the proof.
\end{proof} 

\begin{proof}[Proof of Result \ref{thm:graceb}]
The proof is essentially a generalization of the proof given in Section II-A of \cite{brac}, which was restricted for the particular case where $d=d_y$ for all $y$. We know from the above \textit{lemma} that the optimal encoding and decoding functions are deterministic. Thus, this can be written in a functional form as 
\be \label{Ex}
E(x_1\cdots x_n) = m \ \text{ if } p_a(m|x)=1 ,
\ee 
and
\be \label{Dy}
D_y(m) = b_y \ \text{ if } \ p_b(b_y|y,m)=1 .
\ee 
Here, $E(x_1\cdots x_n)$ is a function whose domain is the set of inputs $x = x_1 \cdots x_n$ and range is the set of messages  $[d]$.
Also, $D_y(m)$ is a function whose domain is the set of messages $[d]$ and range is the set $[d_y]$. We say the decoding strategy is `identity decoding', denoted by $\{\Tilde{D}_y\}$, if
\be \label{Dyt}
\forall y , \ \Tilde{D}_y(m) = m .
\ee 
We want to show that, without loss of generality we can take $\{\Tilde{D}_y\}$ for the maximum success probability. Consider an encoding $E(x)$ \eqref{Ex} and a decoding $\{D_y\}$ \eqref{Dy} that may not be $\{\Tilde{D}_y\}$, that is, there may exists $y$ such that $D_y(m)\neq m $. Let $D_y^{\leftarrow}(b_y)$ be the preimage of $b_y$, that is, $D_y^{\leftarrow}(b_y) = \{ m \in [d] : D_y(m)=b_y \}$.

Subsequently, we consider the following quantity 
\be 
D^{\leftarrow}_1(b_1) \cdots D^{\leftarrow}_n(b_n) = \{ m_1\cdots m_n : D_1(m_1)=b_1, \cdots, D_n(m_n) = b_n \} ,
\ee 
which is simply the set of dit-string $\{m_1\cdots m_n\}$ that is mapped to the sting $b_1 \cdots b_n$.
We define another encoding function $\{\Tilde{E}_x\}$ as follows
\be \label{Ext}
\Tilde{E}(D^{\leftarrow}_1(x_1) D^{\leftarrow}_2(x_2) \cdots D^{\leftarrow}_n(x_n)) = m \ \text{ if } E(x_1\cdots x_n) = m .
\ee 
The above definition of $\tilde{E}$ is not complete since it is not defined if $x_i \notin [d]$ since $D_y^{\leftarrow}(x_i) \in [d] $. In those cases, we take any encoding strategy. Now, we note that $\tilde{E}$ is a well-defined encoding function. Also note that $\tilde{E}$ is a valid encoding for the random access codes considered by us only if $d \leqslant \min_y d_y$. This is because, if  $d > d_y$ for some $y$, then the domain of $\tilde{E}$ may have a string of $n$ dits that does not belong to $x$.

Suppose, for any input pair $x_1\cdots x_n,y$ so that the encoding $E$ and decoding $\{D_y\}$ guesses the correct dit $x_y$. Hence, if the encoding strategy is given by, $E(x = x_1 \cdots x_n) = m$, then the decoding strategy is given by, $D_y(m) = b_y = x_y$. Therefore, we have $D^{\leftarrow}_y(x_y) = m$. As a consequence, the new encoding $\Tilde{E}_x$ and the `identify decoding'  $\{\Tilde{D}_y\}$ also provides the correct answer for at least one input pair from $\{D^{\leftarrow}_1(x_1) D^{\leftarrow}_2(x_2) \cdots D^{\leftarrow}_n(x_n)\},y$. Hence, the average success probability for the  strategy consisting of the encoding $\Tilde{E}_x$ and the `identify decoding'  $\{\Tilde{D}_y\}$ is greater than or equal to that for the strategy with encoding $E$ and decoding $\{D_y\}$. Therefore, we can consider `identity decoding' without loss of generality. 

Next, from Eq.(\ref{simSCg3}), the expression for $S_c$ pertaining to the random access codes for `identity decoding' can be written as
\begin{align}
    S_c = \frac{1}{n \prod_y d_y} \sum_{x} \max_m \left(\sum_{y} P(b_y = x_y|y,m) \right) = \frac{1}{n \prod_y d_y} \sum_{x} \max_m \left(\sum_{y} \delta_{x_y,m}\right),
\end{align}
and for the `identity decoding', the best encoding can be determined   from \eqref{besten} as follows
\begin{align}
    p_a(m^*|x) =1, \, \, \, \, \text{where} \, \, \,  \sum_{y} \delta_{x_y,m^*} \geq \sum_{y} \delta_{x_y,m} \, \,  \, \forall \, \, m \in [d].
\end{align}
Hence, the best encoding pertaining to the `identity decoding' is just sending the dit that belongs to $[d]$ and occurs maximum times in the input string $x_1\cdots x_n$.  

Finally, we provide an expression for $S_c$ for the best classical strategy derived above. In an input string $x_1\cdots x_n$, say, the dit $i$ occurs $n_i$ number of times. The maximum value of a dit can be $\max_y d_y$. Alice sends message $m$ such that $n_m = \max_{i=1,\cdots,d} n_i$. As a result, out of $n$ different values of $y$, they get success $(\max_{i=1,\cdots,d} n_i)$ times. As the total number of dits is $n$, the set of values of $n_i$ should satisfy 
\be \label{sumni}
\sum_{i=1}^{d_{\max} } n_i = n ,
\ee 
where $d_{\max} = \max_y d_y$.
Moreover, dit $i$ may not belong to all $[d_y]$ and thus, $n_i$ can not take all the solutions of the above equation.
Say, $k_i$ is the number of sets among $[d_1],\cdots,[d_n]$ such that dit $i\in [d_y]$. Therefore, we are only interested in those solutions where $n_i \leqslant k_i$.

Given such a solution of $\{n_i\}$, there will be many possible numbers of input dit strings $x$ having that $\{n_i\}$. Next, let us evaluate the number of input dit strings $x$ that can have an arbitrary $\{n_i\}$. In any input string, at most $k_{d_{\max}}$ number of input dits can have the value $d_{\max}$. 
In the given set of input dit strings having $\{n_i\}$, the dit $d_{\max}$ occurs $n_{d_{\max}}$ number of times. Hence, the dit $d_{\max}$ can be arranged in $C^{k_{d_{\max}}}_{n_{d_{\max}}}$ different possible ways. Next, in any input string, at most $k_{d_{\max}-1}$ number of input dits can have the value $(d_{\max}-1)$. 
However, among these $k_{d_{\max}-1}$ number of input dits, $n_{d_{\max}}$ number of dits have already taken the value $d_{\max}$ in the case of the given set of input strings. Also,  in the given set of input dit strings having $\{n_i\}$, the dit $(d_{\max}-1)$ occurs $n_{d_{\max}-1}$ number of times. Therefore, for any of the above-mentioned arrangements of the dit $d_{\max}$, the dit $(d_{\max}-1)$ can be arranged in $C^{k_{d_{\max}-1} - n_{d_{\max}}}_{n_{d_{\max}-1}}$ different possible ways. 
Proceeding in this way, it can be shown that an arbitrary dit $j$  can be arranged in $C^{\alpha_j}_{n_j}$ different possible ways with $\alpha_j = k_j - \sum_{i=j+1}^{d_{\max}} n_i$ for any arrangement of the dits- $d_{\max}$, $(d_{\max}-1)$, $\cdots$, $j+1$. Therefore, given any $\{n_i\}$, there will be $\left(\prod_{j=1}^{d_{\max}}  C^{\alpha_j}_{n_j} \right)$  (with $\alpha_j = k_j - \sum_{i=j+1}^{d_{\max}} n_i$) number of input dit strings having that $\{n_i\}$. Combining these facts we obtain Eq. \eqref{SCex}.
\end{proof}

\section{Derivation of Eq. (\ref{simSCex}) 
and Eq. (\ref{3_measurements})} \label{app3}

From Result \ref{thm:graceb}, 
we can write the following for $n=2$, $d_y=\od$ for all $y$, and $d\leqslant \od$,
\be \label{SCexapp}
S_c(2,\od,d) = \frac{1}{2 \od^2} \sum_{\{n_i\}\in \mathcal{S}} \left[ N_{\{n_i\}} \, \max_{i=1,\cdots,d}\{n_i\} \right]
\ee 
where $N_{\{n_i\}}$ is the number of input dit strings having a given $\{n_i\}$; and $\mathcal{S}$ denotes the set of $\{n_i\}$ satisfying 
\be 
\sum_{i=1}^{\od} n_i = 2
\label{conapp}
\ee 
such that $n_i \leqslant 2$ for all $i$.

Next, let us characterize the set $\mathcal{S}$. It can be noted that there are the following two types of $\{n_i\} \in \mathcal{S}$:  

\begin{enumerate}
    \item For each $i \in [\od]$, $n_i=2$ and $n_j=0$ for all $j\neq i$ and $j \in [ \od]$. 
    
    There are $\od$ number of such $\{n_i\} \in \mathcal{S}$. However, $\max_{i=1,\cdots,d}\{n_i\} = 0$ for each of those $\{n_i\} \in \mathcal{S}$ satisfying $n_i=2$ for any $i$ such that $i \in \{d+1, \cdots, \od\}$ and $n_j=0$ for all $j \in [\od]$ and $j\neq i$. Hence, only $d$ number of $\{n_i\} \in \mathcal{S}$ belonging to this class contribute to  the sum of (\ref{SCexapp}). It is straightforward to check that for each of these $d$ number of $\{n_i\} \in \mathcal{S}$, $N_{\{n_i\}}=1$ and $\max_{i=1,\cdots,d}\{n_i\}=2$. 
    
    \item For each $i,j \in [\od]$ with $i\neq j$, $n_i=n_j=1$ and $n_k=0$ for all $k \notin \{i,j\}$ with $k \in [\od]$. 
    
    There are $C^{\od}_2$ number of such $\{n_i\} \in \mathcal{S}$. However, $\max_{i=1,\cdots,d}\{n_i\} = 0$ for each of those $\{n_i\} \in \mathcal{S}$ satisfying $n_i=n_j=1$ for any $i,j$ with $i \neq j$, $i,j \in \{d+1, \cdots, \od\}$ and $n_k=0$ for all $k \in [\od]$, $k\neq i$, $k \neq j$. There are $C^{(\od-d)}_2$ number of such $\{n_i\} \in \mathcal{S}$ satisfying this. Hence, only $C^{\od}_2 - C^{(\od-d)}_2$ number of $\{n_i\} \in \mathcal{S}$ belonging to this second class  contribute to the sum of (\ref{SCexapp}). It can be checked that for each of these $C^{\od}_2 - C^{(\od-d)}_2$ number of $\{n_i\} \in \mathcal{S}$, $N_{\{n_i\}}=2$ and $\max_{i=1,\cdots,d}\{n_i\}=1$. 
\end{enumerate}

Therefore, we have from Eq. (\ref{SCexapp}) 
\begin{align} 
S_c(2,\od,d) &= \frac{1}{2 \od^2} \left[ 2 d + 2 \left(C^{\od}_2 - C^{(\od-d)}_2\right) \right] \nonumber \\
& = \frac{1}{2 \od^2} \left[  d + 2 d \od - d^2 \right].
\end{align}


Similarly, following the same analysis as above, we can get the expression for $n=3$, $d_y=\od$ for all $y$, and $d\leqslant \od$,
\be \label{sc_3}
S_c(3,\od,d) = \frac{1}{3 \od^3} \sum_{\{n_i\}\in \mathcal{S}} \left[ N_{\{n_i\}} \, \max_{i=1,\cdots,d}\{n_i\} \right]
\ee 
where $N_{\{n_i\}}$ is the number of input dit strings with a given $\{n_i\}$; and $\mathcal{S}$ denotes the set of $\{n_i\}$ satisfying 
\be 
    \sum_{i=1}^{\od} n_i = 3
    \label{conapp4}
\ee 
such that $n_i \leqslant 3$ for all $i$. Now there are three cases that satisfy Eq. (\ref{conapp4}): 

\begin{enumerate}

    \item For each $i \in [\od]$, $n_i=3$ and $n_j=0$ for all $j\neq i$ and $j \in [ \od]$. 
    
  There are $\od$ number of such $\{n_i\} \in \mathcal{S}$. However, $\max_{i=1,\cdots,d}\{n_i\} = 0$ for each of those $\{n_i\} \in \mathcal{S}$ satisfying $n_i=3$ for any $i$ with $i \in \{d+1, \cdots, \od\}$ and $n_j=0$ for all $j \in [\od]$ and $j\neq i$. Hence, only $d$ number of $\{n_i\} \in \mathcal{S}$ belonging to this class contribute to  the sum of (\ref{sc_3}). For each of these $d$ number of $\{n_i\} \in \mathcal{S}$, we have that $N_{\{n_i\}}=1$ and $\max_{i=1,\cdots,d}\{n_i\}=3$.  Hence, the contribution to the sum is $ 3 d$.
    
    \item For each $i,j,k \in [\od]$ with $ i\notin\{j,k\}$, $j \notin \{i,k\}$, $k \notin \{i,j\}$, $ n_{i}=n_{j}=n_{k}=1$ and $ n_{l} = 0$ for all $ l $ $\notin\{i,j,k\}$ and $l \in [\od]$. 
    
        There are $C^{\od}_3$ number of such $\{n_i\} \in \mathcal{S}$. Moreover, $\max_{i=1,\cdots,d}\{n_i\} = 0$ for each of those $\{n_i\} \in \mathcal{S}$ satisfying $n_i=n_j=n_k=1$ for any choice of $i,j,k$ with  $i,j,k \in \{d+1, \cdots, \od\}$, $ i\notin\{j,k\}$, $j \notin \{i,k\}$, $k \notin \{i,j\}$ and $n_l=0$ for all  $l \in [\od]$ and $ l $ $\notin\{i,j,k\}$. There are $C^{(\od-d)}_3$ number of such $\{n_i\} \in \mathcal{S}$ satisfying this. Thus, only $C^{\od}_3 - C^{(\od-d)}_3$ number of $\{n_i\} \in \mathcal{S}$ belonging to this class  contribute to the sum of (\ref{sc_3}). It can be checked that for each of these $C^{\od}_3 - C^{(\od-d)}_3$ number of $\{n_i\} \in \mathcal{S}$, $N_{\{n_i\}}=3!$ and $\max_{i=1,\cdots,d}\{n_i\}=1$. Therefore, the contribution to the sum will be $ (3!)\left ( C^{\od}_3 - C^{\od - d}_3 \right)$. 
    
    \item For each $i,j \in [\od]$ with $ i\neq j$, $ n_{i}=2$, $n_{j}=1$ and $ n_{k} = 0$ for all $k \notin \{i,j\}$ with $k \in [\od]$.
    
    The feasible solutions of Eq.(\ref{conapp4}) that contribute to Eq.(\ref{sc_3}) are of two types: 
    
    \textbf{(A)} $ i  \in [d] $ and $ j \in [\od] - \{i\}$. The number of possible such $\{n_i\} \in \mathcal{S}$ is given by, $ d(\od -1 ) $. Also, for each such $\{n_i\}$, we have that  $N_{\{n_i\}}= 3$ and $\max_{i=1,\cdots,d}\{n_i\} = 2$. Therefore, the contribution to the sum appearing in Eq. (\ref{sc_3}) by this case is $ 6 \, d(\od - 1)$.  
    
    {\bf (B)} $ i \in \{d+1, \cdots, \od\}$ and $j \in [d]$. The number of possible such $\{n_i\} \in \mathcal{S}$ is given by, $d(\od -d )$. And for each such $\{n_i\}$, we have that  $N_{\{n_i\}}= 3$ and $\max_{i=1,\cdots,d}\{n_i\} = 1$. Hence, the contribution to the sum  appearing in Eq. (\ref{sc_3}) for this case is given by, $ 3 \, d(\od - d) .$

    Therefore, the total contribution to the sum of Eq. (\ref{sc_3}) is given by, $ 6 \, d(\od - 1) + 3 \,d(\od - d) = 3\, d (3\, \od - d - 2)$.

\end{enumerate}

    Therefore, we have
\begin{align} 
    S_c(3,\od,d) & = \frac{1}{3\od^3}\left[(3!)\left ( C^{\od}_3 - C^{\od - d}_3 \right ) + 3\, d (3\, \od - d - 2)  + 3 d\right]  \nonumber \\
    & = \frac{d}{3\od^3} \left(d^2 - 1 + 3\od (\od + 1 - d) \right).
\end{align}

\section{Proof of Result \ref{thm:322}} \label{app:322}

Let us take three arbitrary orthonormal bases $\{|\psi_1^1\rangle, |\psi_2^1\rangle\}$, $\{|\psi_1^2\rangle, |\psi_2^2\rangle\}$, and $\{|\psi_1^3\rangle, |\psi_2^3\rangle\}$ in $\mathbbm{C}^2$ such that $M_{x_y|y} = |\psi_{x_y}^y\rangle\!\langle \psi_{x_y}^y|$ with $x_y \in [2]$ for all $y \in \{1,2,3\}$. A unitary can always be applied to these three measurements. Therefore, without any loss of generality, we can assume that 
\begin{align}
    &|\psi_{x_1}^1\rangle\!\langle \psi_{x_1}^1| = \frac{1}{2} \left[\I + (-1)^{x_1} \sigma_z \right] \, \, \, \, \text{with} \, \, x_1 \in [2], \\
    &|\psi_{x_2}^2\rangle\!\langle \psi_{x_2}^2| = \frac{1}{2} \left[\I + (-1)^{x_2} \left( \alpha \sigma_z + \sqrt{1-\alpha^2} \sigma_x \right) \right] \, \, \, \, \text{with} \, \, x_2 \in [2], \\
    &|\psi_{x_3}^3\rangle\!\langle \psi_{x_3}^3| = \frac{1}{2} \left[\I + (-1)^{x_3} \left(  \beta \sigma_z + \gamma \sqrt{1-\beta^2}  \sigma_x  \pm \sqrt{1-\beta^2} \sqrt{1-\gamma^2}  \sigma_y \right) \right]  \, \, \, \, \text{with} \, \, x_3 \in [2].
\end{align}
where $-1 \leq \alpha, \beta, \gamma \leq 1$.

Due to the same reasoning as for $\eqref{maxev}$, the maximum average success probability for the above-mentioned given set of three rank-one projective qubit measurements is given by, 
\be \label{maxevapp322}
S_c(n = 3,\od = 2,d=2)=\frac{1}{24} \ \sum_{ x_{1}, x_{2}, x_{3} =1}^{2} ||M_{x_{1}|1} + M_{x_{2}|2}+ M_{x_{3}|3}||.
\ee

By definition, $||M_{x_{1}|1} + M_{x_{2}|2}+ M_{x_{3}|3}||$ is the maximum eigenvalue of $(M_{x_{1}|1} + M_{x_{2}|2}+ M_{x_{3}|3})$, which can be evaluated easily. Subsequently, it can be checked that 
\begin{align}
 \sum_{ x_{1}, x_{2}, x_{3} =1}^{2}   ||M_{x_{1}|1} + M_{x_{2}|2}+ M_{x_{3}|3}|| =12 &+ \sqrt{3 + 2 \alpha - 2 \beta - 2 \alpha \beta - 2 \gamma \sqrt{1-\alpha^2} \sqrt{1-\beta^2}} \nonumber \\
 & + \sqrt{3 - 2 \alpha + 2 \beta - 2 \alpha \beta - 2 \gamma \sqrt{1-\alpha^2} \sqrt{1-\beta^2}} \nonumber \\
 &+ \sqrt{3 - 2 \alpha - 2 \beta + 2 \alpha \beta + 2 \gamma \sqrt{1-\alpha^2} \sqrt{1-\beta^2}} \nonumber \\
 &+ \sqrt{3 + 2 \alpha + 2 \beta + 2 \alpha \beta + 2 \gamma \sqrt{1-\alpha^2} \sqrt{1-\beta^2}}.
 \label{expfull}
\end{align}
We have found out the minimum of the above expression (\ref{expfull}) by performing numerical optimization. It turns out that 
\begin{align}
    \min_{\alpha, \beta, \gamma \in [-1,1]} \left(\sum_{ x_{1}, x_{2}, x_{3} =1}^{2}   ||M_{x_{1}|1} + M_{x_{2}|2}+ M_{x_{3}|3}||  \right) = 18.
\end{align}
 In other words,
 \begin{align}
    \min_{\alpha, \beta, \gamma \in [-1,1]} \left( \xi_1 + \xi_2 + \xi_3 + \xi_4  \right) = 6,
    \label{minapp}
\end{align}
where 
\begin{align}
    &\xi_1 = \sqrt{3 + 2 \alpha - 2 \beta - 2 \alpha \beta - 2 \gamma \sqrt{1-\alpha^2} \sqrt{1-\beta^2}}, \nonumber \\
    &\xi_2 = \sqrt{3 - 2 \alpha + 2 \beta - 2 \alpha \beta - 2 \gamma \sqrt{1-\alpha^2} \sqrt{1-\beta^2}}, \nonumber \\
    &\xi_3 = \sqrt{3 - 2 \alpha - 2 \beta + 2 \alpha \beta + 2 \gamma \sqrt{1-\alpha^2} \sqrt{1-\beta^2}}, \nonumber \\
    &\xi_4 = \sqrt{3 + 2 \alpha + 2 \beta + 2 \alpha \beta + 2 \gamma \sqrt{1-\alpha^2} \sqrt{1-\beta^2}}. \nonumber 
\end{align}
In order to prove Result \ref{thm:322}, 
it is sufficient to show that $\left( \xi_1 + \xi_2 + \xi_3 + \xi_4  \right) = 6$ only if the three projective measurements are compatible, or $(\alpha,\beta,\gamma ) = \{ (\pm 1/2,\pm 1/2,-1),(\pm 1/2,\mp 1/2,1)\} $. \\
Since $-1 \leq \alpha, \beta, \gamma \leq 1$, we divide the regions of $\alpha$, $\beta$ and  $\gamma$ into the following sub-regions:
\begin{enumerate}[label=\roman*.]
\centering
    \item $\alpha, \beta, \gamma \in [0,1]$,
    \item $\alpha \in [-1,0]$; $\beta, \gamma \in [0,1]$,
    \item $\alpha, \beta \in [-1,0]$; $\gamma \in [0,1]$,
    \item $\alpha, \beta, \gamma \in [-1,0]$,
    \item $\alpha, \gamma \in [0,1]$; $\beta \in [-1,0]$,
    \item $\alpha, \gamma \in [-1,0]$; $\beta \in [0,1]$,
    \item $\alpha, \beta \in [0,1]$; $\gamma \in [-1,0]$,
    \item $\alpha \in [0,1]$; $\beta, \gamma \in [-1,0]$.
\end{enumerate}
We start by considering the above-mentioned sub-region (i), i.e., $\alpha, \beta, \gamma \in [0,1]$. In this case, we note the following holds from numerical evaluation,
\begin{align}
 \min_{\alpha, \beta, \gamma \in [0,1]} \left( \xi_1 + \xi_2  \right) = 2,   
 \label{appmin1}
\end{align}
and 
\begin{align}
 \min_{\alpha, \beta, \gamma \in [0,1]} \left( \xi_3 \right) \geqslant \min_{\alpha, \beta \in [0,1]} \sqrt{3 - 2 \alpha - 2 \beta + 2 \alpha \beta} = 1,  
 \label{appmin2}
\end{align}
since the derivative of the above expression is zero at $\alpha =1$ or/and $\beta=1$.

Next, we evaluate the maximum as well as minimum of $(\xi_1+\xi_2+\xi_3)$ numerically under the constraint that $\left( \xi_1 + \xi_2 + \xi_3 + \xi_4  \right) = 6$. It is obtained that 
\begin{align}
   \min_{\alpha, \beta, \gamma \in [0,1]} \left( \xi_1 + \xi_2  +\xi_3 \right) = \max_{\alpha, \beta, \gamma \in [0,1]} \left( \xi_1 + \xi_2  +\xi_3 \right) = 3, \, \, \, \, \text{when} \, \,  \xi_1 + \xi_2 + \xi_3 + \xi_4  = 6.
\end{align}
Therefore, we have that
\begin{align}
    \xi_1 + \xi_2  +\xi_3 = 3, \, \, \, \, \text{when} \, \, \xi_1 + \xi_2 + \xi_3 + \xi_4   = 6.
    \label{appmin4}
\end{align}
Hence, the following is implied from (\ref{appmin1}), (\ref{appmin2}), (\ref{appmin4}),
\begin{align}
    \xi_1 + \xi_2  = 2, \, \, \, \,  \xi_3=1, \, \, \, \, \text{and} \, \, \, \, \xi_4 = 3, \, \, \, \, \text{when} \, \, \, \, \xi_1 + \xi_2 + \xi_3 + \xi_4   = 6.
\end{align}

Next, it can be checked that $\xi_3 = 1$ only if $\alpha =1$ or/and $\beta=1$.
Now, when $\alpha=1$, then $\xi_4 = 3$ implies that $\beta=1$. Similarly, when $\beta=1$, then $\xi_4 = 3$ implies that $\alpha=1$. 
Therefore, when $\alpha, \beta, \gamma \in [0,1]$, then $\left( \xi_1 + \xi_2 + \xi_3 + \xi_4  \right) = 6$ holds only if $\alpha = \beta = 1$. 

Next, consider the sub-region (iii), i.e., for $\alpha, \beta \in [-1,0]$; $\gamma \in [0,1]$. We note that if $\alpha \rightarrow - \alpha $ and $\beta \rightarrow -\beta $ then the four expressions $\xi_i$ interchange among themselves as we can readily verify $\xi_1 \rightarrow \xi_2$, $\xi_2 \rightarrow \xi_1$, $\xi_3 \rightarrow \xi_4$, $\xi_4 \rightarrow \xi_3$. 
Thus, following a similar calculation as for sub-region (i), we find that $\xi_1 + \xi_2 + \xi_3 + \xi_4  = 6$ holds only if $\alpha=\beta=-1$.
Similarly, for sub-regions (vi) and (viii), one can show that $\xi_1 + \xi_2 + \xi_3 + \xi_4  = 6$ holds only if $-\alpha=\beta=1$ and $\alpha = -\beta=1$, respectively. \\


Next, let us focus on the sub-region (iv), i.e., when $\alpha, \beta, \gamma \in [-1,0]$. In this case, we obtain the following by performing numerical optimizations,
\begin{align}
&\min_{\alpha, \beta, \gamma \in [-1,0]} \left( \xi_1 + \xi_4 \right) = \max_{\alpha, \beta, \gamma \in [-1,0]} \left( \xi_1 + \xi_4 \right) = 2, \, \, \, \, \text{when} \, \,  \xi_1 + \xi_2 + \xi_3 + \xi_4  = 6, \nonumber \\ 
&\min_{\alpha, \beta, \gamma \in [-1,0]} \left( \xi_2 + \xi_4 \right) = \max_{\alpha, \beta, \gamma \in [-1,0]} \left( \xi_2 + \xi_4 \right) = 2, \, \, \, \, \text{when} \, \,  \xi_1 + \xi_2 + \xi_3 + \xi_4  = 6, \nonumber \\ 
&\min_{\alpha, \beta, \gamma \in [-1,0]} \left( \xi_1 + \xi_3 \right) = \max_{\alpha, \beta, \gamma \in [-1,0]} \left( \xi_1 + \xi_3 \right) = 4, \, \, \, \, \text{when} \, \,  \xi_1 + \xi_2 + \xi_3 + \xi_4  = 6, \nonumber \\ 
&\min_{\alpha, \beta, \gamma \in [-1,0]} \left( \xi_2 + \xi_3 \right) = \max_{\alpha, \beta, \gamma \in [-1,0]} \left( \xi_2 + \xi_3 \right) = 4, \, \, \, \, \text{when} \, \,  \xi_1 + \xi_2 + \xi_3 + \xi_4  = 6. \nonumber  
\end{align}
Hence, we can infer that whenever $\xi_1 + \xi_2 + \xi_3 + \xi_4  = 6$,
\be 
 \xi_1 + \xi_4  = \xi_2 + \xi_4  = 2, \quad   \xi_1 + \xi_3  = \xi_2 + \xi_3  = 4 .
\ee 
Therefore, we have $\xi_1 = \xi_2$ if $\xi_1 + \xi_2 + \xi_3 + \xi_4  = 6$. Also, it can be easily checked from the expressions of $\xi_1$, $\xi_2$, $\xi_3$ and $\xi_4$ that 
\begin{align}
    \xi_1^2 + \xi_2^2 + \xi_3^2 + \xi_4^2  = 12.
    \label{newappse}
\end{align}
By putting $\xi_1 = \xi_2 = \xi$, $\xi_3 = 4-\xi$, $\xi_4=2-\xi$, we get from (\ref{newappse}) that
\begin{align}
    2 \xi^2 + (4-\xi)^2 + (2-\xi)^2 = 12.
\end{align}
The possible solutions of the above equation are $\xi=1$ and $\xi=2$.\\
Before proceeding, let us point out the following observations that can be checked numerically,
\begin{align}   
\min_{\alpha, \beta, \gamma \in [-1,0]} \left( \xi_1 \right) = \min_{\alpha, \beta, \gamma \in [-1,0]} \left( \xi_2 \right) = 1.    
\label{minappnewew}
\end{align}
First, we take $\xi=1$. It can be shown that $\alpha, \beta, \gamma \in [-1,0]$, $\xi_1=\xi_2=1$ only if $\alpha=-1$ and $\beta=-1$. 
Next, let us take  $\xi=2$. Consequently, we have that $\xi_1=\xi_2=2$, $\xi_3=2$, $\xi_4=0$. It can be checked that the unique solution of these four equations is given by, $\alpha=-1/2$, $\beta=-1/2$, $\gamma=-1$. 

Next, we remark that for the remaining sub-regions (ii), (v), and (vii) wherein the variables $\alpha, \beta, \gamma$ changes their signs with respect to the sub-region (iv) where $\alpha,\beta, \gamma \in [-1,0]$, the four expressions $\xi_i$ interchange among themselves. Thus, a similar calculation applies to these three regions, and, consequently, the solution for  $\xi_1 + \xi_2 + \xi_3 + \xi_4  = 6$ is the same with the appropriate signs.


Finally, let us note that there are, in general, two cases where we do not observe any advantage. Firstly, $\alpha = \pm 1$ and $\beta = \pm 1$, which implies that the three measurements $\{M_{x_1|1}\}$, $\{M_{x_2|2}\}$ and $\{M_{x_3|3}\}$ are compatible. Secondly, $(\alpha,\beta,\gamma ) = \{ (\pm 1/2,\pm 1/2,-1),(\pm 1/2,\mp 1/2,1)\}$, which are obtained in sub-regions (iv), (ii), (v), and (vii), implies that the three measurements are incompatible.  This completes the proof.



\end{document}